\documentclass[letterpaper, 10 pt, conference]{ieeeconf}
\IEEEoverridecommandlockouts                            
\usepackage{cite}

\usepackage[framemethod=tikz]{mdframed}

\usepackage{hyperref}
\hypersetup{
	colorlinks=true,
	linktoc=all,
	linkcolor=blue, 
	citecolor=blue
}

\usepackage{amsmath,amssymb,amsfonts}
\usepackage{algorithm,algorithmicx,algpseudocode}
\usepackage{graphicx}
\usepackage{subcaption}
\usepackage{mwe}
\usepackage{textcomp}
\usepackage{xcolor}
\definecolor{Blue}{RGB}{88, 105, 225}
\definecolor{QOrange}{RGB}{255, 153, 51}
\definecolor{Orange}{RGB}{255, 153, 50}

\usepackage{dsfont}
\usepackage{bbm}

\DeclareMathOperator{\argmin}{arg\,min}

\DeclareMathOperator*{\diag}{\textit{diag}}
\DeclareMathOperator*{\trace}{\textit{tr}}

\newtheorem{theorem}{\bf Theorem}[section]
\newtheorem{lemma}{\bf Lemma}[section]

\newtheorem{remark}{\bf Remark}[section]

\newtheorem{assumption}{\bf Assumption}[section]
\newtheorem{information}{\bf Information}[section]
\newtheorem{problem}{\bf Problem}[section]

\newcommand{\nnum}{\nonumber}

\newcommand{\LL}{\mathcal{L}}
\newcommand{\RR}{\mathbb{R}}
\newcommand{\II}{\mathbb{I}}
\algnewcommand{\LineComment}[1]{\Statex \(\triangleright\) #1}

\newcommand\oprocendsymbol{\hbox{$\blacksquare$}}
\newcommand\oprocend{\relax\ifmmode\else\unskip\hfill\fi\oprocendsymbol}

\title{\LARGE \bf
Robust Vehicle Lane Keeping Control with Networked Proactive Adaptation}

\author{Hunmin Kim$^{\dagger}$, Wenbin Wan$^{\dagger}$, Naira Hovakimyan$^{\dagger}$, Lui Sha$^{\ddagger}$, and Petros Voulgaris$^{*}$
\thanks{This work has been supported by the National Science Foundation (CNS-1932529).}
\thanks{$^{\dagger}$Hunmin Kim, Wenbin Wan, and Naira Hovakimyan are with the Department of Mechanical Science and Engineering, University of Illinois at Urbana-Champaign, USA.
{\tt\small  \{hunmin, wenbinw2, nhovakim\}@illinois.edu}}%
\thanks{$^{\ddagger}$Lui Sha is with the Department of Computer Science, University of Illinois at Urbana-Champaign, USA.
{\tt\small  lrs@illinois.edu }}%
\thanks{$^{*}$Petros Voulgaris is with the Department of Mechanical Engineering, University of Nevada, Reno, USA.
{\tt\small  pvoulgaris@unr.edu }}
}

\begin{document}
\maketitle
\thispagestyle{empty}
\pagestyle{empty}


\begin{abstract}
Road condition is an important environmental factor for autonomous vehicle control.
A dramatic change in the road condition from the nominal status is a source of uncertainty that can lead to a system failure.
Once the vehicle encounters an uncertain environment, such as hitting an ice patch, it is too late to reduce the speed, and the vehicle can lose control.
To cope with future uncertainties in advance, we study a proactive robust adaptive control architecture for autonomous vehicles' lane-keeping control problems.
The data center generates a prior environmental uncertainty estimate by combining weather forecasts and measurements from anonymous vehicles through a spatio-temporal filter.
The prior estimate contributes to designing a robust heading controller and nominal longitudinal velocity for  proactive adaptation to each new abnormal condition.
Then the control parameters are updated based on posterior information fusion with on-board measurements.
\end{abstract}

\maketitle

\section{Introduction}

Self-driving cars have been one of the most active research areas~\cite{paden2016survey}. 
The autonomous vehicles are a safety-critical system operating in dynamic environments. Their controller should cope with environmental changes robustly.
The current paper is motivated by the safe control problem when hitting an ice patch. Low speed reduces the risk of skidding~\cite{roe1998high}, but it is too late to reduce the speed once hitting an ice patch. The vehicle should slow down in advance and update the controller for the new operating condition to reduce the risk of failure.

The model predictive control (MPC) approach has recently proven its effectiveness for robust and optimal control for vehicle longitudinal and lateral dynamics,~\cite{falcone2007predictive,funke2016collision,williams2018best,carvalho2014stochastic}.
This control strategy also takes such a risk of skidding under dramatic changes of road conditions, because it is still reactive in the sense that it starts to adapt to or learn an uncertain environment after encountering it.
Communication network-enabled controllers can address a part of the problem by incorporating environmental information shared from preceding vehicles,~\cite{na2020disturbance,bertoni2017adaptive}. Motivated by those papers, the current paper leverages vehicle-to-cloud (V2C) communication for a proactive robust adaptive control architecture for lateral dynamics, systematically combining environmental measurements from anonymous vehicles and weather forecasts.


\subsection{Our Contributions}

\begin{figure}[h]
\centering
 \includegraphics[width=\linewidth]{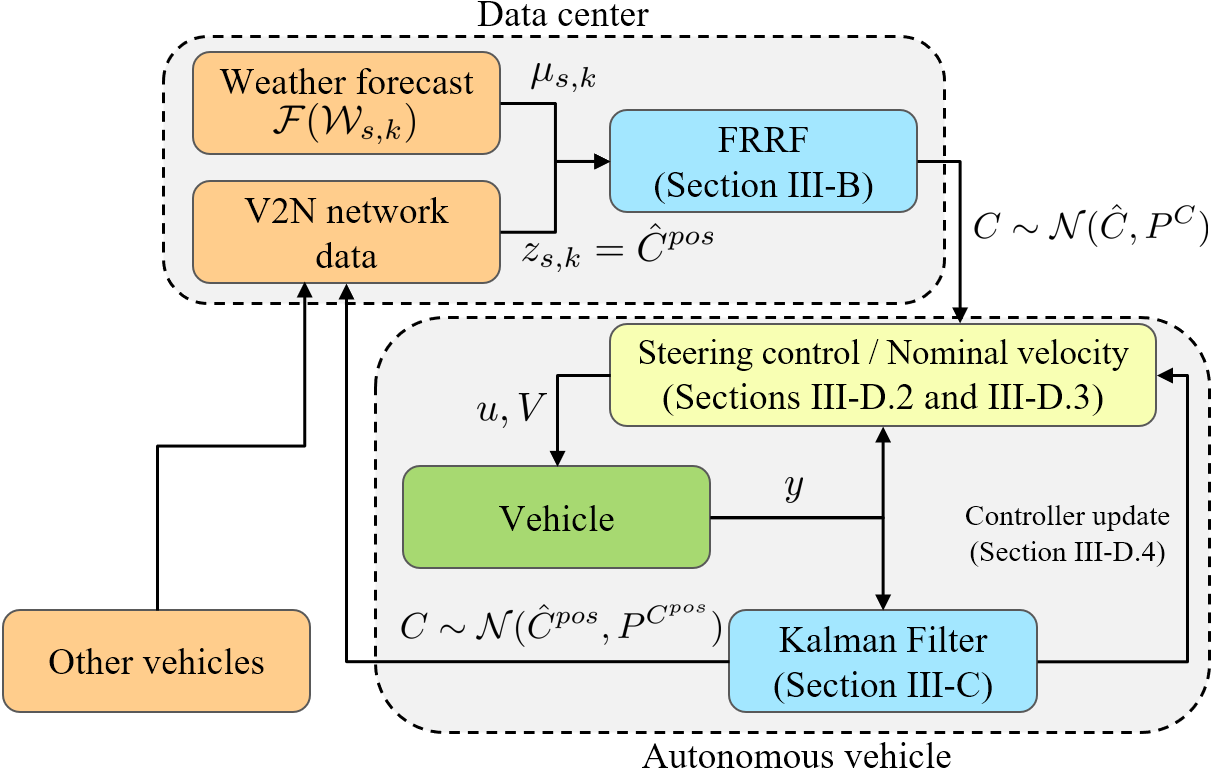}
\caption{(Overall architecture) The data center provides a prior estimate. Controller and velocity are proactively designed based on it for each area of the road.} 
\label{fig:architecture}
\end{figure}

The current paper proposes a novel proactive robust adaptive control architecture for autonomous vehicles to operate in various environmental conditions. Figure~\ref{fig:architecture} illustrates our overall system architecture. The prior of cornering stiffness for multiple areas is estimated by a newly developed fixed rank resilient filter (FRRF) fusing information from the weather forecast and vehicle network data. The $\LL_1$ adaptive heading controller and nominal longitudinal velocity are designed proactively for each area, based on the prior distribution of cornering stiffness.
The proactive adaptation will reduce long-term and large scale uncertainty, while the $\LL_1$ adaptive controller deals with fine-scale uncertainty. Then, based on the posterior distribution of cornering stiffness obtained from the on-board measurements, control parameters are updated.
We also analyze the practical exponential stability of FRRF, when measurements of each area are obtained as a Poisson arrival process.

\subsection{Related Work}


MPC is the most widely used controller for vehicle longitudinal and lateral dynamics.
Reference~\cite{falcone2007predictive} studies a short term predictive active front steering using the MPC for double lane change problem on snow road for nonlinear dynamics. An MPC-enabled tracking controller with varying time-steps has been developed in~\cite{funke2016collision}, where short term fine steps allow stability, and long term coarse steps are dedicated to collision avoidance. In~\cite{williams2018best}, the best-response iteration method is combined with information-theoretic MPC for agile vehicle maneuvers operating in close proximity. Reference~\cite{carvalho2014stochastic} proposes a stochastic MPC integrating Bayesian estimation and introduces a parameterized risk factor that can balance a trade-off between risk and performance.
In~\cite{kamal2012model}, the vehicle measures traffic information to compute the fuel-optimal vehicle control input using MPC, predicting the preceding vehicle's future states.

As an extension of the works mentioned earlier, some approaches deal with uncertain and changing environments adopting a learning algorithm. Reference~\cite{hewing2019cautious} develops a cautious MPC that improves the nominal model via Gaussian process regression and propagates Gaussian uncertainty in time to enhance the performance safely. References~\cite{rosolia2017learning,rosolia2017robust} develop learning MPC, improving performance by learning model from the data, while guaranteeing safety by constructing a safety set based on the previous iteration.

The majority of existing communication network-enabled controllers rely on short-range vehicle-to-vehicle communication between limited connected vehicles for longitudinal control (platooning). 
In particular,~\cite{na2020disturbance} designs a disturbance observer-based controller for platooning. The observer estimates the road slope and shares this information with the following platoons to help them  reduce the fuel consumption. Reference~\cite{bertoni2017adaptive} studies the optimal trade-off between air drag reduction and powertrain energy losses, exploiting a preview from the preceding vehicle. Other network-enabled applications include collision avoidance between vehicles~\cite{miller2002adaptive,biswas2006vehicle,hafner2013cooperative}.

Our control strategy relies on the $\LL_1$ adaptive control,~\cite{cao2008design,hovakimyan2010L1}, which can promptly compensate for unmodeled uncertainties within the filter bandwidth while guaranteeing transient and steady-state performance. 
Due to such merit, the $\LL_1$ adaptive controller has been used as an ancillary controller to ensure that the real system performs as the nominal system. The $\LL_1$ adaptive controller has been integrated with model predictive path integral control~\cite{pravitra2020l1} and with contraction control~\cite{lakshmanan2020safe,gahlawat2020mathcal}.
The $\LL_1$ adaptive controller has been applied to the vehicle lateral dynamics that demonstrates  successful compensation for unmodeled uncertainty, such as parametric uncertainty, wind gust, and disturbances of various natures~\cite{shirazi2017mathcal}.

The prior estimation of the cornering stiffness is a spatio-temporal data fusion problem because the road information contains its attribute as well as spatial and temporal information.
Spatio-temporal modeling and filtering have been widely used in environmental process estimation~\cite{evensen2009data}.
These methodologies' main idea is to model spatial and temporal random effect in dynamic systems and recursively estimate the target variable.
Since these methodologies consider both spatial and temporal correlation for a large scale system, they provide a smooth geostatistical mapping.
Recent research focuses on reducing the computational complexity for potentially massive datasets~\cite{hunt2007efficient,cressie2010fixed,nguyen2014spatio}.
In particular, the spatio-temporal fixed rank filter in~\cite{cressie2010fixed} improves the computational efficiency using spatio-temporal models defined on a fixed dimensional space.
The current paper extends the spatio-temporal fixed rank filter to capture model uncertainty and (unmodeled) biased noises.

\section{Vehicle lateral dynamics and problem statement}

The bicycle model is a simplified vehicle model that has been widely used and has been proven as a good approximation~\cite{carvalho2014stochastic,falcone2007predictive,rajamani2011vehicle}.
Consider Figure~\ref{fig:bicycle}, in which the
variables $p^y$, $p^\psi$, $V$, and $\delta$ denote the lateral position, yaw angle, (longitudinal) velocity, and front steering angle, respectively. Parameters $C_f$, $C_r$, $m$, $I_z$, $\ell_f$, and $\ell_r$ are the
front/rear cornering stiffness, mass, yaw moment of inertia, and distance of front/rear tire from the center of gravity, respectively.

\begin{figure}[thpb]
\centering
 \includegraphics[width=0.42\textwidth]{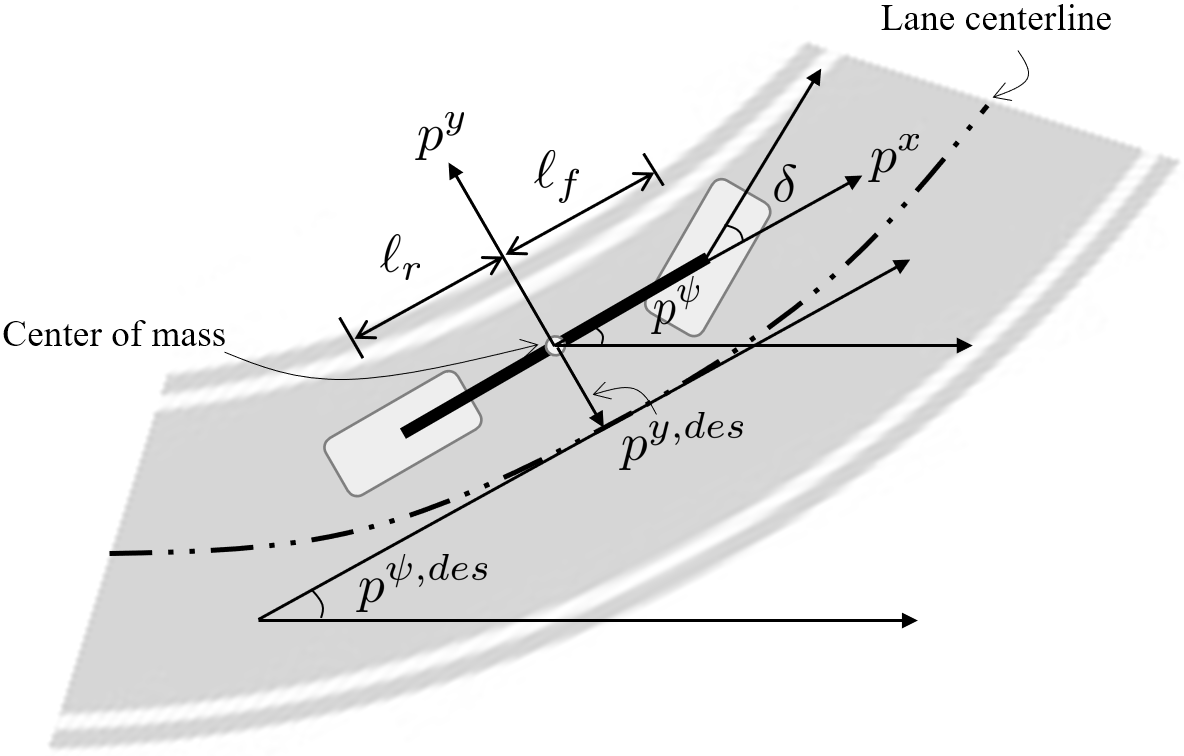}
\caption{Vehicle lateral dynamics.}
\label{fig:bicycle}
\end{figure}

Given a constant velocity $V$, the dynamics of the collective state $p = [p^y, \dot{p}^y, p^\psi, \dot{p}^\psi]^\top$ with heading input $u = \delta$ are described by ((2.31) in~\cite{rajamani2011vehicle})
\begin{align}
    \dot{p} = A^o p + b^o u,
    \label{eq:dynamic_ori}
\end{align}
where the system matrices are
\begin{align*}
    &A^o=\left[
    \begin{array}{cccc}
        0& 1& 0& 0\\
        0& -2\frac{C_f+C_r}{m V}& 0& -V - 2\frac{C_f\ell_f-C_r\ell_r}{m V}\\
        0& 0& 0& 1\\
        0& -2\frac{C_f\ell_f-C_r\ell_r}{I_z V}&0& -2\frac{C_f\ell_f^2+C_r\ell_r^2}{I_z V}
    \end{array}
    \right]\nnum\\
    &b^o =\left[
    \begin{array}{cccc}
        0&
        \frac{2C_f}{m}&
        0&
        \frac{2 C_f\ell_f}{I_z}\\
    \end{array}
    \right]^\top.
\end{align*}
Given the desired lateral position $p^{y,des}$ (center of the lane) and the desired yaw angle $p^{\psi,des}$, the bicycle model~\eqref{eq:dynamic_ori} can be reformulated as error dynamics ((2.45) in~\cite{rajamani2011vehicle}):
\begin{align}
    \dot{x} = A(V,C_f,C_r) x + b(C_f) u + g(V,C_f,C_r) \dot{p}^{\psi,des}, 
    \label{eq:dynamic}
\end{align}
where $x = [x_1,\dot{x}_1, x_2,\dot{x}_2]^\top$, $x_1 \triangleq p^y - p^{y,des}$ and $x_2 \triangleq p^\psi - p^{\psi,des}$ are the error states.
The rate of the desired yaw angle is found by $\dot{p}^{\psi,des} = \frac{V}{R}$, where $R$ is the radius of the road.
The system matrices are
\begin{align*}
    &A(V,C_f,C_r)=\nnum\\
    &\left[
    \begin{array}{cccc}
        0& 1& 0& 0\\
        0& -2\frac{C_f+C_r}{m V}& 2\frac{C_f+C_r}{m}& 2\frac{-C_f\ell_f+C_r\ell_r}{m V}\\
        0& 0& 0& 1\\
        0& -2\frac{C_f\ell_f-C_r\ell_r}{I_z V}& 2\frac{C_f\ell_f-C_r\ell_r}{I_z}& -2\frac{C_f\ell_f^2+C_r\ell_r^2}{I_z V}
    \end{array}
    \right]
    \end{align*}\begin{align*}
    & b(C_f) = b^o, \quad
    g(V,C_f,C_r) =\left[
    \begin{array}{c}
        0\\
        -2\frac{C_f\ell_f - C_r\ell_r}{m V} - V\\
        0\\
        -2\frac{C_f\ell_f^2 + C_r\ell_r^2}{I_z V}\\
    \end{array}
    \right].
\end{align*}
It is worth emphasizing that matrices $A$ and $g$ depend on velocity $V$, and that the matrices $A$, $b$, and $g$ depend on cornering stiffnesses $C_f$ and $C_r$. For notational simplicity, we  express them as $A(V)$, $b$, and $g(V)$, when their dependency on cornering stiffnesses does not need to be emphasized.

The cornering stiffness $C_f$ (and $C_r$) is the coefficient related to the lateral force $F_{f}$ and sliding angle $\beta$, where it has a linear relation $F_f = C_f \beta$ for small $\beta$. This parameter is closely related to the road friction. In this paper, the cornering stiffnesses $C_f$ and $C_r$ are assumed to be unknown, and we can estimate them using the following information.
\begin{information}
i) weather forecast, ii) vehicle network data (anonymous vehicles' cornering stiffness estimates), iii)
on board measurement (GPS, IMU).
\label{infor1}
\end{information}

We formulate the problem of interest as follows.

\begin{problem}
Given Information~\ref{infor1}, the problem is to develop a robust control architecture that stabilizes the error dynamics~\eqref{eq:dynamic} of a vehicle operating in multiple areas through controlling the heading $u=\delta$ and designing the longitudinal velocity $V$.
\label{sec:prob}
\end{problem}

\section{Proactive robust adaptive control}\label{sec:pro}

\subsection{Overall architecture}\label{sec:pro}

The current section provides an overview of the proposed proactive control architecture (Figure~\ref{fig:architecture}) to address Problem~\ref{sec:prob}.
The data center acquires the first two pieces in Information~\ref{infor1} about $C_f$ and $C_r$, which are fused to estimate the prior of cornering stiffness for each area by FRRF algorithm in Section~\ref{sec:prior}.
This will provide a heatmap of cornering stiffness for multiple areas.
Given the cornering stiffness's prior distribution, we pre-design a robust controller and a constant longitudinal velocity for each area of interest.
The design procedure for the $\LL_1$ adaptive control and velocity design for a single area is outlined in Sections~\ref{sec:bound} and~\ref{sec:L1vehicle}, and this procedure can be repeated for other areas.
The on-board measurement and prior estimate are fused to produce a posterior cornering stiffness estimate of the current area in Section~\ref{sec:posterior}. The posterior information is then used to update the control parameters presented in Section~\ref{sec:parameter}.

\begin{remark}
We should design the velocity before encountering an uncertain environment. 
Since the system matrices depend on the velocity, the velocity and controller are simultaneously designed based on the prior distribution of cornering stiffness rather than the posterior distribution.
\end{remark}

\subsection{Prior estimation: Spatio-temporal fixed rank resilient filtering}\label{sec:prior}

The information about cornering stiffness from weather forecasts and anonymous vehicles contains its attribute as well as spatial and temporal information. Therefore, the prior estimation of the cornering stiffness can be formulated as a spatio-temporal data fusion problem.
We propose to extend the spatio-temporal fixed rank filter~\cite{cressie2010fixed} to a fixed rank resilient filter (FRRF) such that the filter captures model uncertainty and (unmodeled) biased noises.
Assume the cornering stiffness $C_f$ (or $C_r$) follows a spatio-temporal process $\{q_{s,k}: s \in D, k \in \{1,2,\cdots, n_D\}\}$, where $q_{s,k} \in \RR$, and $D$ is the index set of spatial domains (or area), and $k$ is the discrete-time index.
Domain $D$ could be finite, or countably infinite.
Now consider the spatio-temporal mixed effect model~\cite{cressie2010fixed,nguyen2014spatio}:
\begin{align}
    q_{s,k} &= \mu_{s,k} + S_{s,k}\eta_k + \xi_{s,k}\nnum\\
    z_{k} & = [z_{s_{1_k},k},z_{s_{2_k},k},\cdots,z_{s_{{n_k}},k}]^\top\nnum\\
    z_{s,k} &= q_{s,k} + \epsilon_{s,k},
    \label{eq:spatio}
\end{align}
where $z_{s,k} \in \RR$ is an output of area $s$ at time $k$ and is subject to measurement noise $\epsilon_{s,k}$.
At time $k$, we observe $n_k$ sensor outputs, and
the collection of outputs is denoted by $z_{k} \in \RR^{n_k}$. The collection of measured area indices is denoted by $O_k = \{s_{1_k},s_{2_k},\cdots,s_{n_k}\} \subseteq D$.
In the current paper, anonymous vehicles estimate the cornering stiffness of the presence area and send the estimates to the center through V2C communication. This vehicle network data represents the output $z_{k}$.

Consider the first equation in~\eqref{eq:spatio}. The first term $\mu_{s,k} \in \RR$ is a known time-varying value that models large scale variation.
For the cornering stiffness estimation problem, we assume $\mu_{s,k}$ is a function of weather forecast ${\mathcal{W}}_{s,k}$ (including temperature, precipitation, humidity, wind, and more), i.e., $\mu_{s,k}={\mathcal{F}}({\mathcal{W}_{s,k}})$. The mapping function ${\mathcal{F}}(\cdot)$ can be found by standard learning/regression algorithms (e.g., Gaussian process regression, neural network, basis function regression) by using historical input-output data $\langle \mathcal{W}_{s,k}, C_{f,k}\rangle$.
In the current paper, we assume the function ${\mathcal{F}}$ is given.

The second term $S_{s,k}\eta_k$ captures a smooth small scale variation that correlates the spatial relation between different areas by the finite $n_\eta$-dimensional spatial basis $S_{s,k}$.
Matrix $S_{s,k}$ is known, but the state variable $\eta_k \in \RR^{n_\eta}$ is unknown. The third term $\xi_{s,k} \in \RR$ presents time-dependent fine-scale variation that captures the nugget effect. The state variable $\eta_k$ is supposed to evolve according to the following dynamic equation:
\begin{align}
    \eta_{k+1} = H_k \eta_k + G_k d_k + \zeta_k,
    \label{eq:eta}
\end{align}
where $H_k$ and $G_k$ are known matrices.
The first term $H_k \eta_k$ captures temporal correlation, and the row of $H_k$ can be chosen to be zeros, if the corresponding component $\eta_{k+1}$ does not change dynamically. The second term $G_k d_k$ denotes a biased noise and model uncertainty, where $d_k \in \RR^{n_d}$ is unknown, and it can be seen as an unknown input.
This term is absent in~\cite{cressie2010fixed,nguyen2014spatio}. The last term $\zeta_k \in \RR^{n_\eta}$ represents a fine-scale variation of hidden state $\eta_k$.
All noises $\epsilon_{s,k}$, $\xi_{s,k}$, $\zeta_k$
are independent zero-mean Gaussian with covariance $P_{s,k}^\epsilon$, $P_{s,k}^\xi$, and $P_{k}^\zeta$, respectively.

This paper extends the fixed rank filtering to FRRF, incorporating biased noise and model uncertainty $d_k$, described in~\eqref{eq:eta}. Our interest is to recursively estimate the hidden state $q_{s_*,k}$ for the query area $s_* \in D$. 

Denote $\mu_k$, $S_k$, $\epsilon_k$, $\xi_k$  the collection of the corresponding values for all $s \in O_k$ and define $P_k^\epsilon = \diag (P_{s,k}^\epsilon)$ and $P_k^\xi = \diag (P_{s,k}^\xi)$ for all $s \in O_k$ for simplicity.
The matrix $E_k \in \{0,1\}^{n_k \times n_D}$ denotes the output matrix having $1$ for $(1,s_{1_k}), \cdots, (n_k,s_{n_k})$ elements, and $0$ for the others.
Let $\hat{v}_k$, $\tilde{v}_k \triangleq v_k-\hat{v}_k$, and $P_k^{v} \triangleq {\mathbb{E}}[(v_k-\hat{v}_k)(v_k-\hat{v}_k)^\top]$ denote the estimate, estimation error, and estimation covariance of a variable $v$ at time $k$.

The estimate $\hat{q}_{s_*,k}$ represents an estimate of cornering stiffness $C_f$ (or $C_r$). 
Appendices~\ref{app:FRRFder}, and~\ref{app:FRRFpro} present detailed derivation and properties of FRRF.
The derivation of the algorithm is motivated by fixed rank filtering~\cite{cressie2010fixed} and simultaneous unknown input and state estimation algorithms~\cite{gillijns2007unbiased,wan2019attack}, and, thus, they also share similar properties. 
In particular, the proposed algorithm is the best linear unbiased estimation (Lemma~\ref{lemma:BLUE}), and the estimation error is practically exponentially stable when measurements for each area are obtained as a Poisson process (Theorem~\ref{the:FRRFstability}).

The summary of the proposed algorithm is shown below. Given the output $z_{s,k}$ and the previous estimate $\hat{\eta}_{k-1}$, the unknown variable $\eta_k$ in~\eqref{eq:eta} is estimated by rejecting the unmodeled uncertainty $d_k$.
The variable $q_{s,k}$ in~\eqref{eq:spatio} is estimated from $\hat{\eta}_{k}$ compensating for fine-scale variation $\xi_{s,k}$ by its estimate $\hat{\xi}_{s,k}$.\\
\textbf{Recursive prediction:}
\begin{align}
    &\hat{\eta}_{k|k-1} = H_{k-1}\hat{\eta}_{k-1}+G_{k-1}M_k(z_k -\mu_{k} - S_{k}H_{k-1}\hat{\eta}_{k-1})\nnum\\
    &P_{k|k-1}^\eta = (\II-G_{k-1} M_kS_k)H_{k-1}P_{k-1}^\eta H_{k-1}^\top(\II\nnum\\
    &-G_{k-1} M_kS_k)^\top +G_{k-1} M_k(P_k^\epsilon+E_kP_k^\xi E_k^\top)M_k^\top G_{k-1}^\top \nnum\\
    &+ (\II-G_{k-1} M_kS_k)P_{k-1}^\zeta(\II-G_{k-1} M_kS_k)^\top,
    \label{eq:pred}
\end{align}
where
\begin{align}
    M_k = ( G_{k-1}^\top S_k^\top R_k^{-1}S_kG_{k-1} )^{\dagger} G_{k-1}^\top S_k^\top R_k^{-1},
    \label{eq:M}
\end{align}
and $R_k=S_k(H_{k-1}P_{k-1}^\eta H_{k-1}^\top+P_{k-1}^\zeta) S_k^\top+P_k^\epsilon + E_kP_k^{\xi}E_k^\top$.\\
\textbf{Recursive estimation:}
\begin{align}
    \hat{\eta}_k &= \hat{\eta}_{k|k-1} + K_k (z_k - \mu_{k}-S_{k}\hat{\eta}_{k|k-1})\nnum\\
    P_k^\eta &= (\II-K_kS_k)P_{k|k-1}^\eta(\II-K_kS_k)^\top \nnum\\
    &+ K_k (P_k^\epsilon+E_kP_k^\xi E_k^\top)K_k^\top \nnum\\
    &+ (\II-K_kS_k)M_k(P_k^\epsilon+E_kP_k^\xi E_k^\top)K_k^\top \nnum\\
    &+ K_k(P_k^\epsilon+E_kP_k^\xi E_k^\top)M_k^\top (\II-K_kS_k)^\top,
    \label{eq:estim}
\end{align}
where
\begin{align}
   K_k=(P^{\eta}_{k|k-1} S_{k}^\top - M_{k}  (P_k^\epsilon+E_kP_k^\xi E_k^\top))\tilde{R}_k^{-1} ,
   \label{eq:K}
\end{align}
and
$\tilde{R}_k=S_{k} P^{\eta}_{k|k-1} S_{k}^\top + (P_k^\epsilon+E_kP_k^\xi E_k^\top)-S_{k} M_{k}(P_k^\epsilon+E_kP_k^\xi E_k^\top)-(P_k^\epsilon+E_kP_k^\xi E_k^\top) M_{k}^\top S_{k}^\top$.\\
\textbf{Estimation of $q_{s_*,k}$:} 
\begin{align*}
    &\hat{q}_{s_*,k} = \mu_{s_*,k} + S_{s_*,k} \hat{\eta}_k + \hat{\xi}_{s_*,k}\nnum\\
    &\hat{\xi}_{s_*,k} = 
         L_{s_*,k} (z_{k}^{s_*} - \bold{1}\mu_{s_*,k} - \bold{1}S_{s_*,k}\hat{\eta}_{k|k-1})
         \end{align*}\begin{align}
&P_{s_*,k}^q=S_{s_*,k}K_kS_k P_{k|k-1}^\eta(S_{s_*,k}K_kS_k)^\top+P_{s_*,k}^\epsilon L_{s_*,k}L_{s_*,k}^\top\nnum\\
&+S_{s_*,k}K_k(P_k^\epsilon+E_kP_k^\xi E_k^\top) (S_{s_*,k}K_k)^\top\nnum\\
&+S_{s_*,k}K_kP_k^{s,s_*,\epsilon}L_{s_*,k}^\top
+L_{s_*,k}^\top P_k^{s_*,s,\epsilon}K_k^\top S_{s_*,k}^\top\nnum\\
&+S_{s_*,k}K_kS_k G_{k-1}M_k((P_k^\epsilon + E_k P_k^\xi E_k^\top) K_k^\top S_{s_*,k}^\top\nnum\\
&+ P_k^{s,s_*,\epsilon}L_{s_*,k}^\top)+(S_{s_*,k} K_k (P_k^\epsilon + E_k P_k^\xi E_k^\top) \nnum\\
& + L_{s_*,k}P_k^{s_*,s,\epsilon})(S_{s_*,k}K_kS_k G_{k-1}M_k)^\top, {\rm \ if \ } s_* \in O_k\nnum\\
&P_{s_*,k}^q = S_{s_*,k}P_{k}^\eta S_{s_*,k}^\top+P_{s_*,k}^\xi, {\rm \ otherwise},
\label{eq:yhat}
\end{align}
where $z_k^{s_*}$ is the collection of outputs $z_{s,k}$ for $s =s^*$,
\begin{align}
    L_{s_*,k} =
    \left\{
    \begin{array}{cc}
             (\bold{1}^\top \bar{R}_{s_*,k}^{-1} \bold{1})^{-1}\bold{1}^\top \bar{R}_{s_*,k}^{-1}&{\rm if } \ s_* \in O_k \\
         0& {\rm otherwise} ,
    \end{array}
    \right.
    \label{eq:G}
\end{align}
and
$\bar{R}_{s_*,k}=P_{s_*,k}^{\epsilon}\II + \bold{1}
S_{s_*,k} P_{k|k-1}^\eta S_{s_*,k}^\top \bold{1}^\top-\bold{1}
S_{s_*,k}G_{k-1}M_kP_k^{s,s_*,\epsilon}
-P_k^{s_*,s,\epsilon}(\bold{1}S_{s_*,k}G_{k-1}M_k)^\top$, $P_k^{s,s_*,\epsilon}\triangleq {\mathbb E}[\epsilon_k(\epsilon_k^{s_*})^\top]$.

\subsection{Posterior estimation: Real-time local information fusion}\label{sec:posterior}

The posterior estimation of cornering stiffness utilizes the standard Kalman filtering on the bicycle model~\cite{bevly2006integrating,baffet2009estimation}, taking three outputs: 1) GPS measurement, 2) IMU measurement, and 3) prior estimate obtained in Section~\ref{sec:prior}. This method will use the prior estimate as one of the (less frequently measured) outputs. GPS and IMU sensors that are already implemented or easy to be installed are enough to estimate the cornering stiffness~\cite{baffet2009estimation}. Although it does not require additional devices, additional GPS provides a better estimation result~\cite{bevly2006integrating}. We refer to~\cite{sierra2006cornering} for various filtering algorithms with their pros and cons.

\subsection{$\LL_1$ adaptive control with proactive velocity design}\label{sec:L1big}

We implement the $\LL_1$ adaptive controller~\cite{hovakimyan2010L1} for the lane-keeping control, which provides rapid disturbance compensation within the filter bandwidth, while guaranteeing transient and steady-state performance. Different controllers should be designed for different areas, because the prior distribution of the cornering stiffness varies by location. The current section provides a controller design for one area $s$, and the same design procedure can be repeated for all other areas of interest.

Section~\ref{sec:L1} introduces the $\LL_1$ adaptive controller on its nominal system~\cite{hovakimyan2010L1}. Section~\ref{sec:bound} discusses how to transform the error dynamics~\eqref{eq:dynamic} to the nominal system for the $\LL_1$ adaptive controller using the prior distribution of the cornering stiffness. In particular, the nominal system model for the $\LL_1$ adaptive controller is determined by the mean of the prior distribution obtained in Section~\ref{sec:prior}, and a $95\%$ confidence interval of the uncertainty bounds.
Section~\ref{sec:L1vehicle} provides the design procedure for the $\LL_1$ adaptive controller and the velocity for the error dynamics.

\subsubsection{$\LL_1$ adaptive controller}\label{sec:L1}

Consider the following system:
\begin{align}
    \dot{x}(t) &= A_m x(t) + b_m (w u_{ad}(t)+ \theta^\top x(t)+ \sigma(t)) \nnum\\
    y(t)&=c^\top x(t) \quad \quad x(0)=x_0,
    \label{eq:l1dynamic}
\end{align}
where $A_m$, $b_m$, and $c$ are known system matrices/vectors, and $A_m$ is Hurwitz.
Parameter $w \in \RR$ represents the unknown input gain, and the state-dependent uncertainty is represented by $b_m \theta^\top x(t)$, where $\theta$ is an unknown vector. The uncertain parameters satisfy Assumption~\ref{asm:L1}.
The signal $\sigma(t)$ represents the time-varying external disturbance that satisfies Assumption~\ref{asm:L2}.
\begin{assumption}
We have $w \in \Omega = [w_l,w_u]$, and $\theta \in \Theta$,  where the bound $[w_l,w_u]$ and convex set $\Theta$ are known.
\label{asm:L1}
\end{assumption}
\begin{assumption}
The disturbance signal $\sigma(t)$ is continuously differentiable, and the signal and its derivative are uniformly bounded, i.e.,
$|\sigma(t)|\leq\Delta$, and $|\dot{\sigma}(t)|\leq d_\sigma<\infty$ for $\forall t\geq0$, where the bounds $\Delta$ and $d_\sigma$ are known.
\label{asm:L2}
\end{assumption}
The control input $u_{ad}(t)$ is an adaptive controller that consists of state predictor, adaptation law, and low-pass filter.
In what follows, we describe the $\LL_1$ adaptive controller.


\textbf{State predictor:}
The state predictor is given by
\begin{align*}
    \dot{\hat{x}}(t) &= A_m\hat{x}(t)+b_m(\hat{w}(t)u_{ad}(t)+\hat{\theta}^\top x(t)+\hat{\sigma}(t))\nnum\\
    \hat{y}(t)&=c^\top\hat{x}(t) \quad \quad  \hat{x}(0)=\hat{x}_0.
\end{align*}

\textbf{Adaptation laws:}
The adaptation laws are given by:
\begin{align*}
    \dot{\hat{w}}(t)&=\Gamma Proj(\hat{w}(t),-\tilde{x}^\top(t)Pb_m u_{ad}(t)) &&\hat{w}(0)=\hat{w}_0\nnum\\
    \dot{\hat{\theta}}(t)&=\Gamma Proj(\hat{\theta}(t),-\tilde{x}^\top(t)Pb_mx(t))&&\hat{\theta}(0)=\hat{\theta}_0\nnum\\
    \dot{\hat{\sigma}}(t)&=\Gamma Proj(\hat{\sigma}(t),-\tilde{x}^\top(t)Pb_m)&&\hat{\sigma}(0)=\hat{\sigma}_0,
\end{align*}
where $\tilde{x}(t)=\hat{x}(t)-x(t)$ is the prediction error, and $\Gamma>0$ is an adaptation gain, 
$Proj(\cdot,\cdot)$ is the projection operator defined in Definition B.3 in~\cite{hovakimyan2010L1}. The projection operator guarantees that each estimate remains in its desired domain. Matrix $P$ is a symmetric positive definite matrix, solving the algebraic Lyapunov equation
$A_m P + P A_m^\top = -Q$ for a given symmetric positive definite matrix $Q$.

\textbf{Control law:} The adaptive control input is designed by
\begin{align*}
    u_{ad}(s) = -k D(s)(\hat{\eta}(s)-k_gr(s)),
\end{align*}
where $\hat{\eta}(t)=\hat{w}(t)u_{ad}(t)+\hat{\theta}^\top(t)x(t)+\hat{\sigma}(t)$ and $k_g = -1/(c^\top A_m^{-1}b_m)$, and $k>0$ is a constant.
The signal $r(s)$ is the Laplace transform of the reference signal, and $D(s)$ is a strictly proper transfer function that leads to a strictly proper stable low-pass filter
\begin{align*}
    C(s) = \frac{wkD(s)}{1+wkD(s)}
\end{align*}
with $C(0)=1$. We choose $D(s)=1/s$ in this paper.
We need to choose the controller such that
the $\LL_1$-norm condition is satisfied: $\|G(s)\|_{\LL_1}L<1$, where $G(s)=H(s)(1-C(s))$, $H(s)=(s\II-A_m)^{-1}b_m$, and $L=\max_{\theta \in \Theta}\|\theta\|_1$.
Since $\theta$ is constant and $D(s)=1/s$, the $\LL_1$-norm condition reduces to
\begin{align}
    A_g=
    \left[
    \begin{array}{cc}
         A_m+b_m\theta^\top&b_m w  \\
         -k\theta^\top& -kw
    \end{array}
    \right]
    \label{eq:Ag}
\end{align}
being Hurwitz for all $\theta \in \Theta$ and $w \in \Omega_0$.

\subsubsection{System transformation and bounds of uncertainties}\label{sec:bound}

The system~\eqref{eq:dynamic} is uncertain, where the system matrices $A(V,C_f,C_r)$ and $b(C_f)$ depend on unknown cornering stiffness $C_f$ and $C_r$, while $A_m$ and $b_m$ in~\eqref{eq:l1dynamic} are known.
We will use the mean values $\hat{C}_f = \hat{q}_{s_*,k}$ (and $\hat{C}_r=\hat{q}_{s_*,k}'$ for rear cornering stiffness) of the prior distribution to construct uncontrolled nominal system matrices, i.e., $A(V,\hat{C}_f,\hat{C}_r)$ and $b(\hat{C}_f)$.
Consider the control input $u = u_{m} + u_{ad}$, where we will later choose $u_{m} \triangleq -k_m x$ such that $A_m(V) = A(V,\hat{C}_f,\hat{C}_r)-b_m k_m$ becomes Hurwitz and $b_m \triangleq b(\hat{C}_f)$.
Then, the system~\eqref{eq:dynamic} becomes the nominal system~\eqref{eq:l1dynamic} for the $\LL_1$ adaptive controller, where the following relations approximate the uncertainties:
\begin{align}
    &b(C_f)=b_m w\nnum\\
    &\theta = \frac{1}{w}b_m^\dagger(A(V,C_f,C_r)-A(V,\hat{C}_f,\hat{C}_r)) + k_m (\frac{1}{w}-1)\nnum\\
    &\sigma = b_m^\dagger g \dot{p}^{\psi,des}
    \label{app:model}
\end{align}
with $b_m^\dagger = [0, \frac{m}{4\hat{C}_f},0,\frac{I_z}{4\hat{C}_f\ell_f}]$.

It is required to approximate the bounds of uncertainties $\Theta$, $\Omega$, $\Delta$, and $d_\sigma$ to design the $\LL_1$ adaptive controller. To provide those sets, we assume that the cornering stiffness's actual value is bounded by its $95\%$ confidence interval of the prior distribution $\mathcal{N}(\hat{C}_f,P^{C_f}) = \mathcal{N}(\hat{q}_{s_*,k},P_{s_*,k}^q)$ (or $\mathcal{N}(\hat{C}_r,P^{C_r}) = \mathcal{N}(\hat{q}_{s_*,k}',P_{s_*,k}^{q'})$ for rear cornering stiffness), i.e., given the prior distributions, we have constants $\underline{C}_f,\bar{C}_f, \underline{C}_r$, and $\bar{C}_r$ that
\begin{align}
    C_f \in [\underline{C}_f,\bar{C}_f], \ 
    C_r\in [\underline{C}_r,\bar{C}_r].
\label{app:C}
\end{align}

Assumption~\ref{asm1} is a mild condition, because the typical vehicle model satisfies $\ell_r\geq \ell_f$ and $I_z \geq m$, as shown in~\cite{rajamani2011vehicle}. In Assumption~\ref{asm2}, the first condition implies that the road's curve is bounded, and the second condition implies that the change of the curve is bounded.

\begin{assumption}[Vehicle model] The vehicle model satisfies $\ell_r\geq \ell_f$ and $I_z\frac{\ell_r}{\ell_f}-m\geq0$.
\label{asm1}
\end{assumption}

\begin{assumption}[Radius of road]
We have
$R\geq \underline{R}$ and $|\frac{\dot{R}}{R^2}| \leq \bar{R}_d$ for some $\underline{R},\bar{R}_d > 0$.
\label{asm2}
\end{assumption}

\begin{lemma}
Consider Assumptions~\ref{asm1} and~\ref{asm2}.
Given~\eqref{app:model} and~\eqref{app:C}, the bounds of uncertainties are found by
\begin{align}
    \Omega &= [\frac{\underline{C_f}}{\hat{C}_f},\frac{\bar{C_f}}{\hat{C}_f}]\nnum\\
    \Delta(V) &= \frac{1}{2\hat{C}_f \underline{R}}(2\bar{C}_f\ell_f+\bar{C}_r\ell_r(\frac{\ell_r}{\ell_f}-1)+\frac{mV^2}{2})\nnum\\
    d_{\sigma}(V) &= \frac{\bar{R}_d}{2\hat{C}_f }(2\bar{C}_f\ell_f+\bar{C}_r\ell_r(\frac{\ell_r}{\ell_f}-1)+\frac{mV^2}{2})\nnum\\
    \Theta(V) &= \frac{1}{V}(\Theta_1 \times \Theta_2 \times\Theta_3\times\Theta_4),
    \label{eq:bounds}
\end{align}
where
\begin{align}
    \Theta_1&= k_{m}^{(1)} V \Xi, \quad \Theta_3=k_{m}^{(3)} V \Xi\nnum\\
    \Theta_2 &= [-\frac{(m+I_z)(\bar{C}_f-\hat{C}_f)}{2 m  \underline{C}_f}+\frac{(I_z\frac{\ell_r}{\ell_f}-m)(\underline{C}_r-\hat{C}_r)}{2 m  \bar{C}_f},\nnum\\
    &-\frac{(m+I_z)(\underline{C}_f-\hat{C}_f)}{2m  \bar{C}_f}+\frac{(I_z\frac{\ell_r}{\ell_f}-m)(\bar{C}_r-\hat{C}_r)}{ 2 m  \underline{C}_f}]\nnum\\
    &+k_{m}^{(2)} V \Xi\nnum\\
    \Theta_4 & = [-\frac{(m+I_z)(\bar{C}_f-\hat{C}_f)}{2 m  \underline{C}_f}+\frac{(I_z\frac{\ell_r}{\ell_f}-m)(\underline{C}_r-\hat{C}_r)}{2 m  \bar{C}_f},\nnum\\
    &-\frac{(m+I_z)(\underline{C}_f-\hat{C}_f)}{2 m  \bar{C}_f}+\frac{(I_z\frac{\ell_r}{\ell_f}-m)(\bar{C}_r-\hat{C}_r)}{2 m  \underline{C}_f}]\nnum\\
    &+k_{m}^{(4)} V \Xi
    \label{app:theta}
\end{align}
and $k_{m}^{(i)}$ is the $i^{th}$ element of $k_m$, and $\Xi \triangleq [\frac{\hat{C}_f}{\bar{C}_f}-1, \frac{\hat{C}_f}{\underline{C}_f}-1]$.
\label{lem1}
\end{lemma}

Notice that the sets $\Theta(V)$, $\Delta(V)$, and $d_{\sigma}(V)$ are a function of the velocity $V$.

The proofs of Lemmas can be found in Appendix~\ref{app:proof}.

\subsubsection{$\LL_1$ adaptive controller and nominal velocity design}\label{sec:L1vehicle}

The $\LL_1$ adaptive controller guarantees transient and steady-state performance with respect to the reference system and design system. The reference system is the non-adaptive version of the $\LL_1$ adaptive controller.
The design system is an ideal system that does not depend on the uncertainties. According to Theorem 2.2.2 in~\cite{hovakimyan2010L1}, the performance of the system can be arbitrary close to the reference system ($x_{ref}(t)$ and $u_{ref}(t)$) by increasing the adaptation gain $\Gamma$ without sacrificing robustness.
Lemma 2.1.4 in~\cite{hovakimyan2010L1} analyzes the error between the reference system and the design system ($\|x_{ref}-x_{des}\|_{\LL_{\infty}}$ and $\|u_{ref}-u_{des}\|_{\LL_{\infty}}$),
where its upper bound is proportional to $\|G(s)\|_{\LL_1}$.
The term $\|G(s)\|_{\LL_1}$ can be close to zero by arbitrarily increasing the filter bandwidth $k$. However, this performance improvement trades off with the robustness.
In particular, the time-delay margin decreases to zero, as $k$ increases to infinity. Therefore, we need to design $k_m$, $C(s)$, and $V$ balancing the performance and robustness optimally.


The matrix $A_m(V)$ must be Hurwitz, but it depends both on gain $k_m$ and velocity $V$. To relax this complexity, we propose to use the common Lyapunov function approach. We first design control gains $k_m$ and $P$ such that $A_m(V)$ is Hurwitz for any velocity $V \in [V_{\min},V_{\max}]$,
where $V_{\min}$ and $V_{\max}$ are the minimum and maximum velocity of the area.
Upon that, we  choose the velocity $V$ and filter $C(s)$ simultaneously through an optimization problem.

Given $\hat{C}_f$ and $\hat{C}_r$, we should choose constant vector $k_m$ and symmetric positive definite matrix $P$ such that
\begin{align}
    A_m(V) P + P A_m^\top(V) <0
    \label{eq:lyap}
\end{align}
holds for all $V_{\min} \leq V \leq V_{\max}$.
One does not need to explore the entire domain of $V$, but only need to check the minimum $V_{\min}$ and maximum $V_{\max}$. 

\begin{lemma}
Assume that there exists $0\leq \alpha(V) \leq 1$ such that $A_m(V) = \alpha(V) A_m(V_{\min}) + (1-\alpha(V)) A_m(V_{\max})$ for any $V_{\min} \leq V \leq V_{\max}$.
Then there exists positive definite matrix $Q(V)$ 
such that $A_m(V) P + P A_m^\top(V) =-Q(V)$ for any $V_{\min} \leq V \leq V_{\max}$ if and only if $A_m(V_{\min}) P + P A_m^\top(V_{\min})=-Q_{\min}$, and $A_m(V_{\max}) P + P A_m^\top(V_{\max}) = -Q_{\max}$ for some symmetric positive definite matrices $P$, $Q_{\min}$, and $Q_{\max}$.
\label{lem:Am}
\end{lemma}

For any $V \in [V_{\min},V_{\max}]$, we have
\begin{align*}
A_m(V) &= \alpha(V) A_m(V_{\min}) + (1-\alpha(V)) A_m(V_{\max})
\end{align*}
for $\alpha(V)=\frac{\frac{V_{\min}V_{\max}}{V}-V_{\min}}{V_{\max}-V_{\min}}$. Therefore, by Lemma~\ref{lem:Am}, we can choose $k_m$ and $P$ such that the condition in~\eqref{eq:lyap} holds both for $V_{\min}$ and $V_{\max}$.
The adaptation gain $\Gamma>0$ can be chosen as a very large number to enhance the adaptation performance.

We can choose the filter gain $k$ and the velocity $V$
balancing the performance and robustness.
The performance is characterized by $\|G(s)\|_{\LL_1}$ as in~\cite{li2008filter}. The robustness is characterized by a lower bound of $k$, which prevents the time-delay margin from converging to zero.
The optimization problem can be formulated by
\begin{align}
    &\max_{k,V \in [V_{\min},V_{\max}]} V\nnum\\
    &s.t. \ \|G(s)\|_{\LL_1} \leq \lambda_{gp}, {\rm \ for \ } \forall w \in \Omega\nnum\\
    &\quad \ k \leq \bar{k}
    \label{eq:prog}
\end{align}
for some constants $\bar{k}>0$ and $\lambda_{gp}<\frac{1}{L}$. Recall that that $G(s) = H(s)(1-C(s))$. Given $\lambda_{gp}$, one could find the performance bounds of $\|x_{ref}-x_{des}\|_{\LL_1}$ and $\|u_{ref}-u_{des}\|_{\LL_1}$ in Lemma 7 in~\cite{cao2008design}.

\subsubsection{Real-time controller update}\label{sec:parameter}

It is critically important to ensure that the matrices $A_m$ and $A_g$ are Hurwitz for all possible uncertainties. Given the posterior distribution ${\mathcal{N}}(\hat{C}_f^{pos},P^{C_f^{pos}})$ (or ${\mathcal{N}}(\hat{C}_r^{pos},P^{C_r^{pos}})$ for rear cornering stiffness) from Kalman filter in Section~\ref{sec:posterior}, we can construct the $95\%$ confidence interval of the posterior distribution of $C_f$ and $C_r$. We check online whether $A_g(V)$ is Hurwitz for the new set of uncertainties. If it does not hold, we update $k$ in real-time such that $A_g(V)$ is Hurwitz:
\begin{align*}
k = &\argmin_{k} |k-k_*|\nnum\\
& \ {\rm s.t. } \ A_g(V) \ {\rm being \ Hurwitz},
\end{align*}
where $k_*$ is the current gain.
It is worth to note that $A_m$ does not need to be re-tuned, because it depends only on $\hat{C}_f$ and $\hat{C}_r$, and not on the bounds of uncertainties.
Furthermore, we design  it to be Hurwitz for the entire possible velocity range.
 


\section{Simulation}\label{sec:simul}

The current section demonstrates the performance of the proposed control architecture. In particular, Section~\ref{sec:sim:pri} presents the resilient estimation performance of FRRF in the presence of (biased) unmodeled uncertainty.
Based on the prior estimate, we design the $\LL_1$ adaptive controller for the areas of interest and illustrate the lane keeping performance in Section~\ref{sec:sim:adp}. Lastly, we show a trend of maximum velocity in~\eqref{eq:prog} with respect to changing nominal cornering stiffness in Section~\ref{sec:sim:speed}.
All values are in standard SI units; $m$ (meter) for $\ell_f$, $\ell_r$, $R$, and $x_1$; $rad$ for $\delta$ and $x_2$; $m/s$ for $V$; $N/rad$ for $C_f$ and $C_r$; $kg$ for $m$; $kg \cdot m^2$ for $I_z$.

\subsection{Prior estimation by FRRF}\label{sec:sim:pri}

We consider the square area that is divided into $25$ identical small squares, i.e., $n_D=25$. The ground truth cornering stiffness holds $C_f =  C_r$ for all the areas, although this information is unknown to the control authority.

In this simulation, we will conduct one FRRF algorithm and use the distribution of $q_{s,k}$ to estimate both cornering stiffnesses $C_f$ and $C_r$, i.e., $\hat{C}_f=\hat{C}_r$, $\underline{C}_f=\underline{C}_r$, and $\bar{C}_f=\bar{C}_r$.
Consider the stochastic process $q_{s,k}$ described in~\eqref{eq:spatio} and~\eqref{eq:eta}.
Matrices $S_{s,k}$ is chosen to be the W-wavelets as in~\cite{cressie2010fixed,kwong1994w}. Matrix $H_k$ and $G_k$ are chosen to be identity matrices.
Noises are zero-mean Gaussian with known covariance $P_{s,k}^{\epsilon}=10$, $P_{s,k}^{\xi}=100$, and $P_{k}^{\zeta}=100\II$.
Predictive cornering stiffnesses from the weather forecast are randomly generated by uniform distribution for each time $k$ and each area, i.e. $\mu_{s,k} \sim Unif(C_{ice}, C_{dry})$ for $\forall k,s$, except $s=1,2$, where $C_{ice} = 19000$ and $C_{dry} = 84000$ are the conservative lower bound and upper bound.
We intentionally choose time-invariant $\mu_{1,k}$ and $\mu_{2,k}$ for all $k$ to compare the fine-scale tracking performance.
The data center obtains the measurement $z_{s,k}$ (for each area) as a Poisson distribution with parameter $\lambda=20$. The unmodeled system uncertainty $d_k$ is made up of $d_k = 100\sin(k\pi)$.

Given the initial condition $\hat{\eta}_0 = { \bold{0}}$ with covariance $P_0^\eta = 1000\II$, we conduct FRRF algorithm in Section~\ref{sec:prior}, and present the simulation results in Figures~\ref{fig:heatmap} and~\ref{fig:prior}.
For each time $k$, FRRF generates a heat map for the cornering stiffness. Figure~\ref{fig:heatmap} presents a series of heat maps produced by FRRF algorithm, where the color represents the mean value $\hat{q}_{s,k}$ of the corresponding area $s$. 

Figure~\ref{fig:prior} compares the tracking errors when the outputs are sparsely measured (as a Poisson with $\lambda=20$) and are fully measured at areas $1$ and $2$.
Areas $1$ and $2$ are the left bottom corner and its right cell, respectively.
The estimation errors for all areas remain in their noise level. FRRF algorithm estimates the ground truth cornering stiffness resiliently, where the errors do not depend on the presence of $d_k$, as shown in the first subfigure.
FRRF with the full measurement exhibits an improved tracking performance of fine-scale variation a lot than that with the spares measurement, as presented in the second and third subfigures. This is because FRRF with the full measurement successfully reduces the estimation error by compensating for unmodeled uncertainty at each iteration.
The average trace norm of variance in the whole area is $\trace(P_{k}^{q,full})=1983.7$ with the full measurement and $\trace(P_{k}^{q,\lambda = 20})=3190.4$ with the spares measurement. 

\begin{figure}[!thpb]
    \centering
    \includegraphics[width=\linewidth]{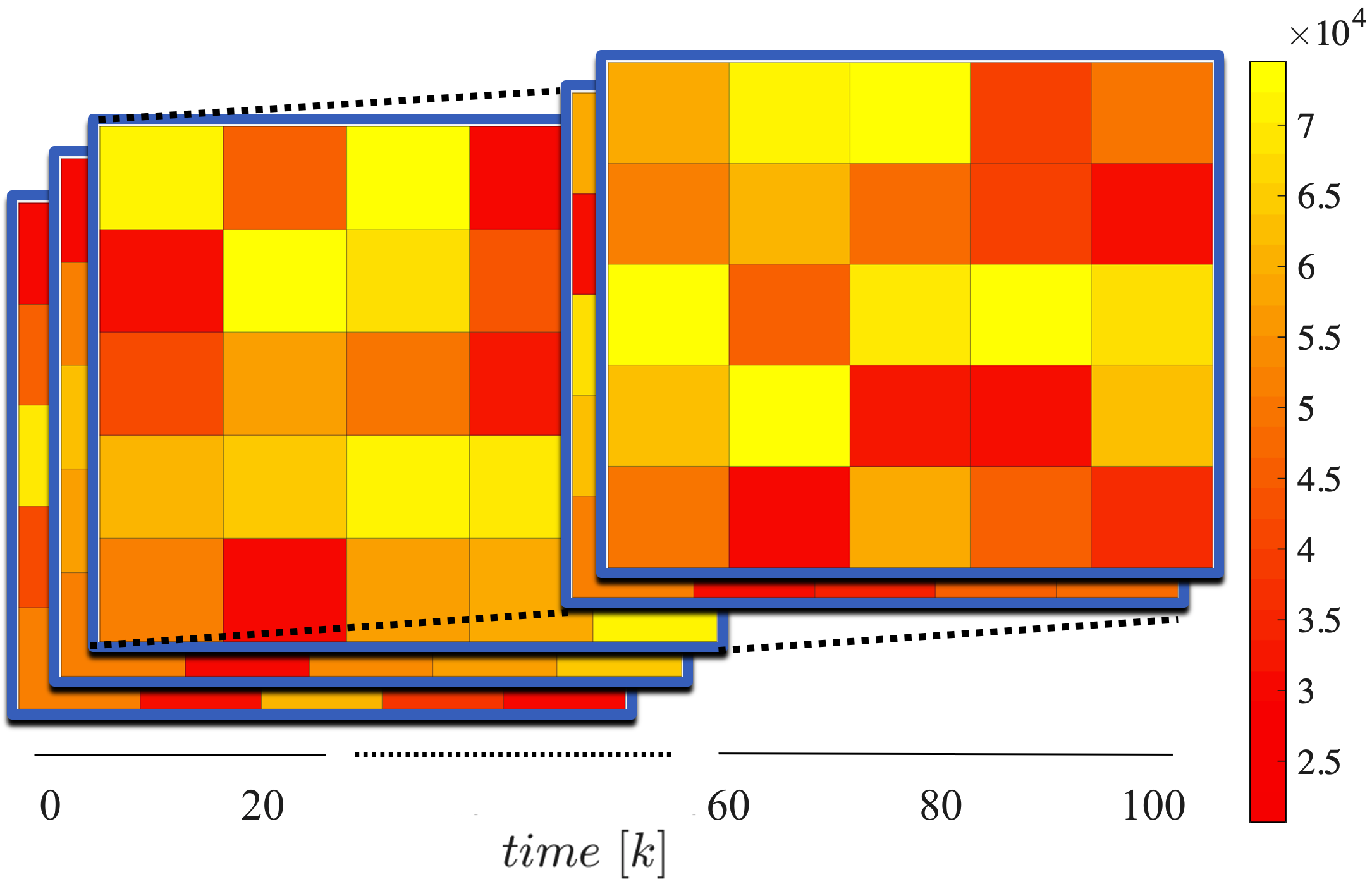}
    \caption{Prior estimation heatmap. The color represents the mean value of the estimate in the corresponding area.}
    \label{fig:heatmap}
\end{figure}

\begin{figure}[!thpb]
    \centering
    \includegraphics[width=\linewidth]{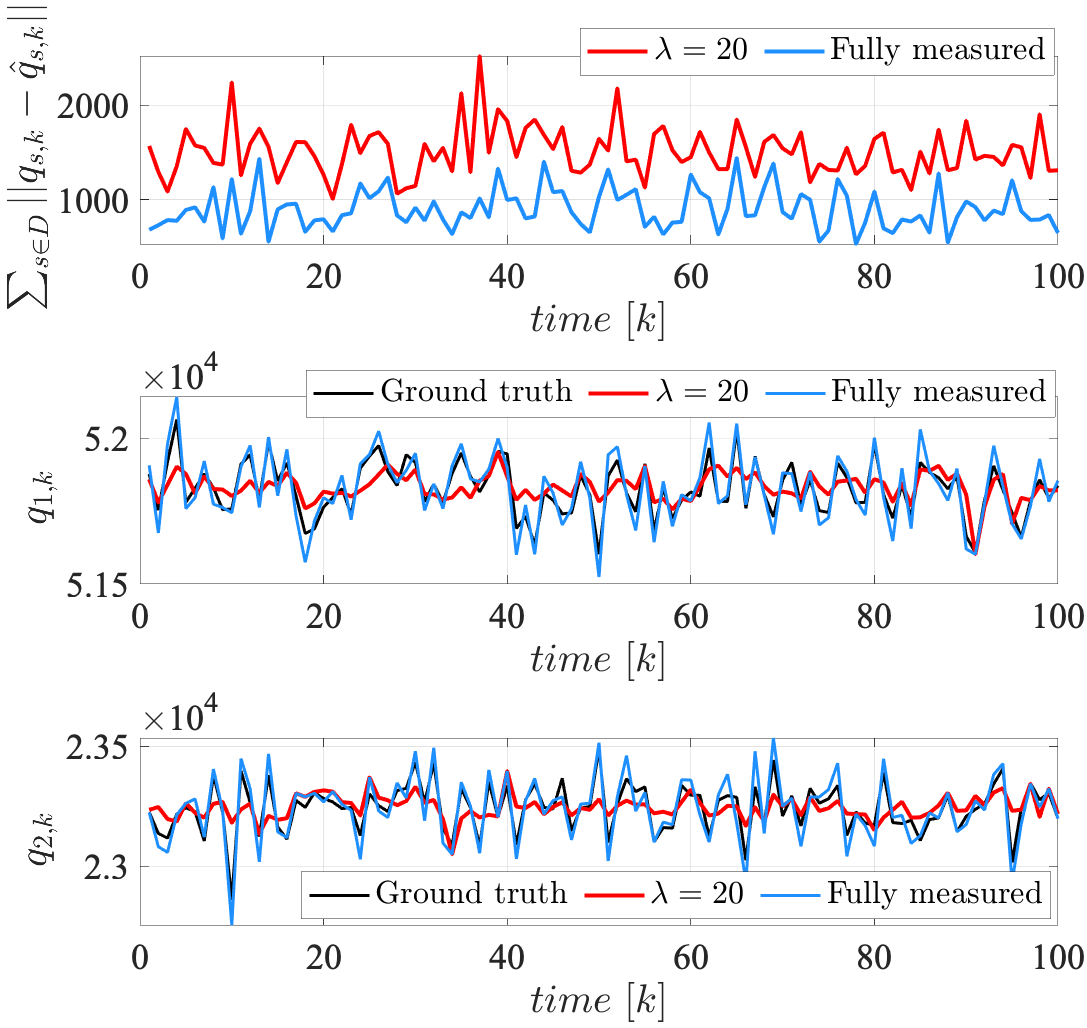}
    \caption{Prior estimation performance; (top) total estimation error; (middle, bottom) ground truth cornering stiffness and estimates with the full measurement and spares measurement at areas $1$ and $2$.}
    \label{fig:prior}
\end{figure}

\subsection{Proactive $\LL_1$ adaptive control}\label{sec:sim:adp}

\begin{figure}[thpb]
    \centering
    \includegraphics[width=\linewidth]{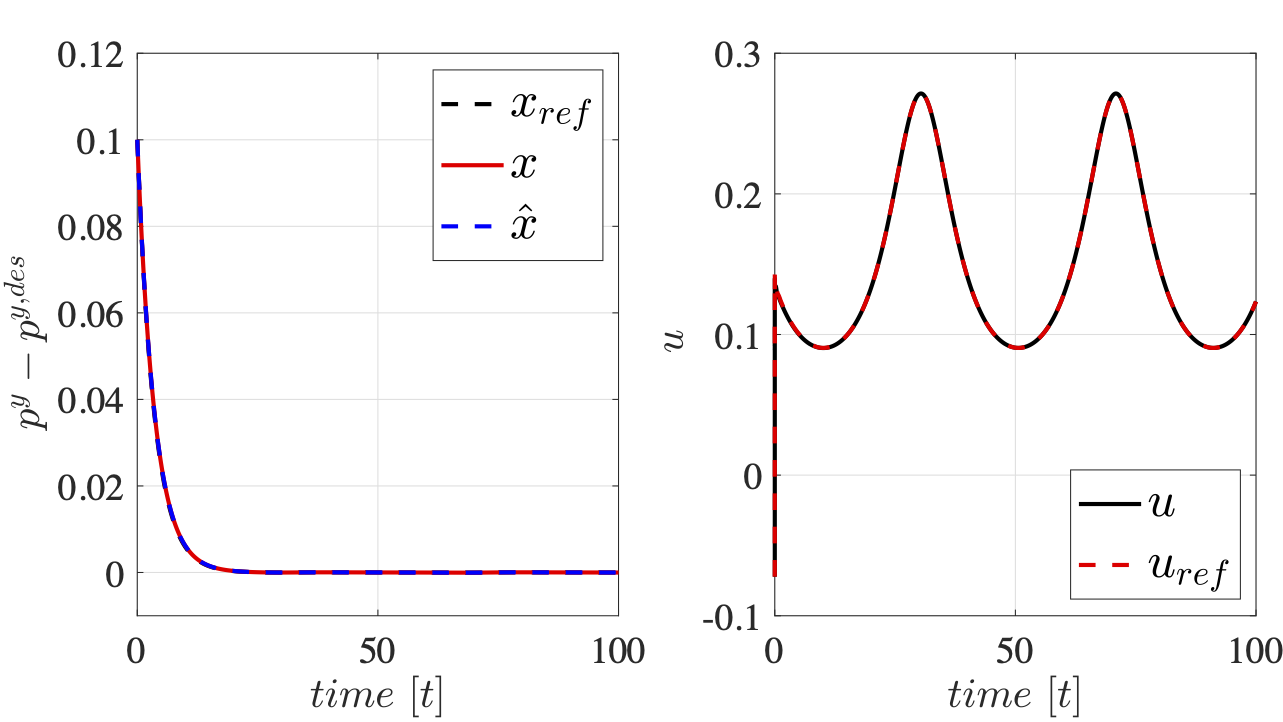}
    \caption{Raining condition. Error states and control inputs in area $1$ ($C_{1,f}=C_{1,r}=51867$).}
    \label{fig:rain}
\end{figure}

\begin{figure}[thpb]
    \centering
    \includegraphics[width=\linewidth]{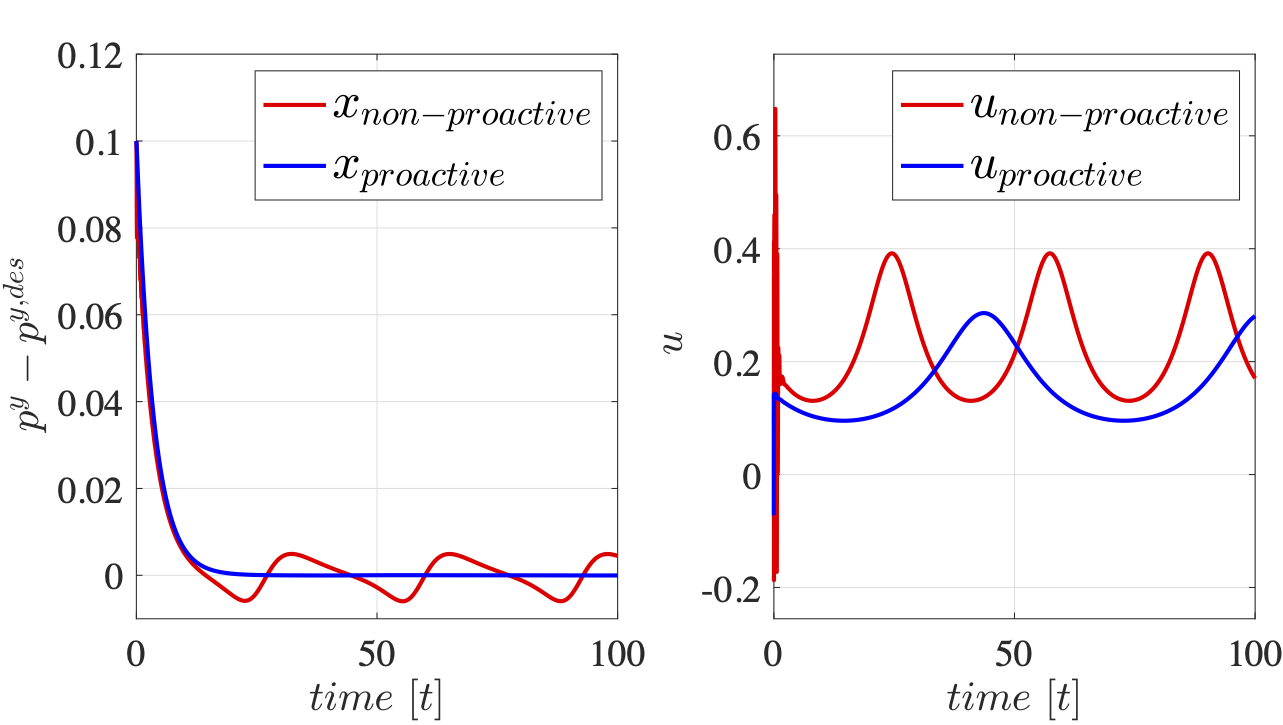}
    \caption{Snowing condition. Error states and control inputs in area $2$ ($C_{2,f}=C_{2,r}=23214$).
    }
    \label{fig:snow}
\end{figure}
The current section compares the tracking performance of the proactive $\LL_1$ adaptive control and a non-proactive version of it. We refer to~\cite{shirazi2017mathcal} to compare
the $\LL_1$ adaptive controller's performance with that of other types of controllers.

The vehicle's system parameters are as follows~\cite{shirazi2017mathcal}:
$m=1573$, $I_z=2873$, $\ell_f=1.1$, $\ell_r=1.58$. To challenge the maneuver, we choose the time-varying radius of the road 
$R(s) = 15\sin(\frac{1}{120}s)+30$ for area $1$ and area $2$, where $s$ is the arc length.
The vehicle operates in areas $1$ and $2$, where the ground truth cornering stiffnesses are $C_{1,f}=C_{1,r}=q_{1,50}=51867$, and $C_{2,f}=C_{2,r}=q_{2,50}=23214$ at time $k=50$.
Controllers are designed based on the prior distribution at time $k = 50$, i.e., $q_{1,k} \sim {\mathcal{N}}(51826,1413)$ and $q_{2,k} \sim {\mathcal{N}}(23240,1937)$. 

Performance bound is chosen to be $\lambda_{gp} = 0.585$, and $\bar{k}=10$. Given the performance bound and distributions for area $1$ and $2$ in Section~\ref{sec:sim:pri}, we design the $\LL_1$ adaptive controller for areas $1$ and $2$, as follows:
\begin{align*}
    &k_{1,m}=k_{2,m}=
    \left[
    \begin{array}{cccc}
         0.7223& 2.5855& -0.6669& 0.1873  \\
    \end{array}
    \right]^\top,
\end{align*}
$k_{1}=k_{2} = 10$, $V_{1}= 18.61$, $V_{2}= 12.96$, $\Gamma_1= \Gamma_2 = 100000$,
\begin{align*}P_{1} =
\left[
\begin{array}{cccc}
    1.9111 &   0.0053 &  0.3485  &  0.0090\\
    0.0053 &   0.0196 & -0.0052  & -0.0294\\
    0.3485 &  -0.0052 &  5.2183  &  0.0438\\
    0.0090 &  -0.0294 &  0.0438  &  0.0543
    \end{array}
    \right],
\end{align*}
and
\begin{align*}P_{2} =
\left[
\begin{array}{cccc}
    1.9064 &  0.0180  &  0.4485  &  0.0483\\
    0.0180 &  0.0636  & -0.0211  & -0.0928\\
    0.4485 & -0.0211  &  3.9649  &  0.1609\\
    0.0483 & -0.0928  &  0.1609  &  0.1834
    \end{array}
    \right].
\end{align*}
For a comparison, we also consider the non-proactive controller for area $2$ designed by:
$k_{np,m}=k_{2,m}$, $k_{np} = k_2$, $V_{np}= 22.96$, $\Gamma_{np} = \Gamma_2$, and
\begin{align*}P_{np} =
\left[
\begin{array}{cccc}
    1.9195 &   0.0153  &  0.3201  &  0.0119\\
    0.0153 &   0.0360  &  0.0095  & -0.0532\\
    0.3201 &   0.0095  &  8.0706  &  0.0908\\
    0.0119 &  -0.0532  &  0.0908  &  0.1015
    \end{array}
    \right].
\end{align*}

Figure~\ref{fig:rain} presents the performance of the proactively designed $\LL_1$ adaptive controller under the raining condition ($C_{1,f}=C_{1,r}=51867$). The controller can successfully stabilize the error dynamics under the changing road radius. With a large adaptation gain, the system performance is arbitrarily close to that of the reference system.

Figure~\ref{fig:snow} compares the proactive $\LL_1$ adaptive controller's tracking performance and the non-proactive version under the snowing condition and changing road radius.
The system with the proactive controller does not have performance degradation compared to operating in the raining condition.
We found that the non-proactive controller designed for dry road conditions (around $C_f=C_r=80000$) failed to stabilize the system.
As discussed before, one could increase $k$ to guarantee stability, but this will harm the robustness.
To illustrate the performance difference between the proactive controller and non-proactive controller without increasing $k$, we consider the controller designed for $C_f=C_r=60000$. The non-proactive controller could also stabilize the error dynamics through compensation of uncertainties, but presents a relatively large error, when the vehicle operates outside of its nominal status.

\subsection{Vehicle velocity curve}\label{sec:sim:speed}

We study a trend of maximum velocity chosen by the optimization problem~\eqref{eq:prog}.
The control parameters, performance bound, and the bound of $k$ remain unchanged throughout the range of cornering stiffness for a fair comparison.
The maximum velocity decreases as the nominal cornering stiffness decreases, as shown in Figure~\ref{fig:CVplot}.
The proposed control architecture slows down the vehicle in advance to guarantee the desired performance and robustness, when the road is expected to be slippery from the prior estimate.

\begin{figure}[!thpb]
    \centering
    \includegraphics[width=\linewidth]{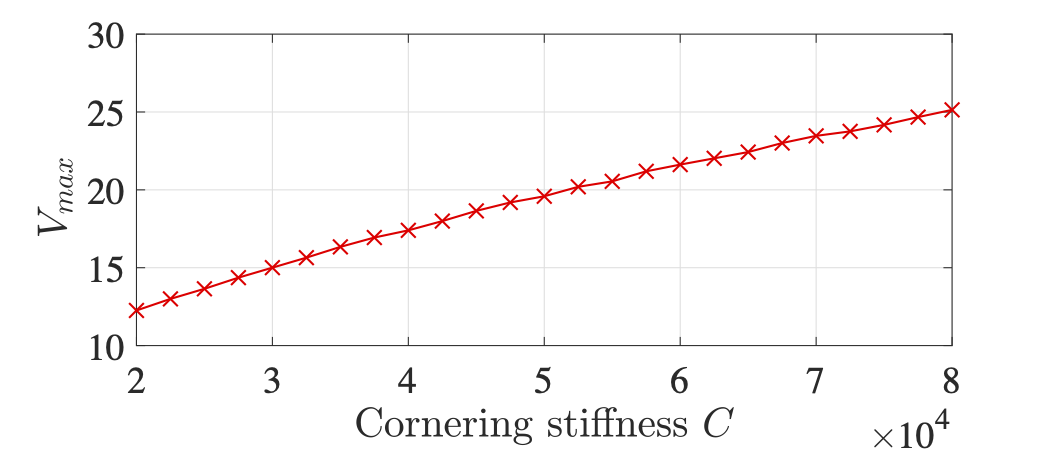}
    \caption{Velocity designed by~\eqref{eq:prog} with respect to changing nominal cornering stiffness $\hat{C}_f=\hat{C}_r$.}
    \label{fig:CVplot}
\end{figure}

\section{Conclusion}
We study a proactive robust adaptive control architecture for autonomous vehicles operating in various environmental conditions.
The weather forecast and vehicle network data are used to estimate the unknown cornering stiffness by newly developed FRRF. Given the prior estimate for multiple areas, the $\LL_1$ adaptive controller and velocity are designed for each road area, balancing the performance and robustness.
The posterior estimate is obtained by combining the prior estimate and the on-board measurement. The controller is updated based on the posterior distribution, if it violates the $\LL_1$-norm condition. The simulation demonstrates that, given the fixed desired performance, the maximum available velocity decreases as the cornering stiffness decreases.

\bibliographystyle{ieeetr}
\bibliography{a_reference}

\section{Appendix}

Section~\ref{app:FRRFder} derives FRRF algorithm, and Section~\ref{app:FRRFpro} shows its properties. Section~\ref{app:proof} presents the omitted proofs of Lemmas~\ref{lem1} and~\ref{lem:Am}.

\subsection{FRRF derivation}\label{app:FRRFder}
\textbf{Prediction of $\eta_k$.}
The previous estimate $\hat{\eta}_{k-1}$ and its covariance $P_{k-1}^\eta$ are given in the last iteration.
Assuming that the estimate $\hat{\eta}_{k-1}$ is unbiased, the uncertainty $d_{k-1}$ can be estimated from the prediction error:
\begin{align}
    \hat{d}&_{k-1}= M_k(z_k -\mu_{k} - S_{k}H_{k-1}\hat{\eta}_{k-1})\nnum\\
    & = M_k (\epsilon_k+ E_k\xi_k+S_k(\zeta_{k-1}+G_{k-1} d_{k-1}+H_{k-1}\tilde{\eta}_{k-1})),
    \label{eq:dhat}
\end{align}
where the error is a function of $M_kS_kG_{k-1} d_{k-1}$.
To provide an unbiased estimate, we will choose $M_k$ later such that $M_kS_kG_{k-1} =\II$. The error dynamics of the uncertainty estimate are
\begin{align}
\tilde{d}_{k-1}=-M_k (\epsilon_k+ E_k\xi_k+S_k(\zeta_{k-1}+H_{k-1}\tilde{\eta}_{k-1})).
\label{eq:dtilde}
\end{align}
Given $\hat{d}_{k-1}$, the current state $\eta_k$ can be predicted by the dynamical system~\eqref{eq:eta}:
\begin{align*}
    \hat{\eta}_{k|k-1} = H_{k-1}\hat{\eta}_{k-1}+G_{k-1}\hat{d}_{k-1}.
\end{align*}
The estimation error $\tilde{\eta}_{k|k-1}= \eta_k-\hat{\eta}_{k|k-1}$ becomes
\begin{align*}
\tilde{\eta}&_{k|k-1}
    = H_{k-1}\tilde{\eta}_{k-1}+G_{k-1} \tilde{d}_{k-1}+\zeta_{k-1}\nnum\\    
    &=(\II-G_{k-1} M_kS_k)H_{k-1}\tilde{\eta}_{k-1}-G_{k-1} M_k (\epsilon_k+E_k\xi_k)\nnum\\
    &+(\II-G_{k-1} M_kS_k)\zeta_{k-1},
\end{align*}
where the relation~\eqref{eq:dtilde} is applied.
These error dynamics induce the covariance update for $P_{k|k-1}^\eta$ in~\eqref{eq:pred}.

\textbf{Estimation of $\eta_k$.} Given $\hat{\eta}_{k|k-1}$, the prediction is corrected by the prediction error $z_k - \mu_{k}-S_{k}\hat{\eta}_{k|k-1}$:
\begin{align*}
    \hat{\eta}_k = \hat{\eta}_{k|k-1} + K_k (z_k - \mu_{k}-S_{k}\hat{\eta}_{k|k-1}).
\end{align*}
The estimation error becomes
\begin{align}
\tilde{\eta}_k &= (\II-K_kS_k)\tilde{\eta}_{k|k-1}-K_k(\epsilon_k+E_k\xi_k),
\label{eq:etatilde}
\end{align}
which results in the covariance $P_k^\eta$ in~\eqref{eq:estim}.

\textbf{Estimation of $q_{s_*,k}$.} Our interest is to estimate the hidden state $q_{s_*,k}$ for the query area $s_*$, which can be estimated by the process model~\eqref{eq:spatio} as follows:
\begin{align*}
    \hat{q}_{s_*,k} = \mu_{s_*,k} + S_{s_*,k} \hat{\eta}_k + \hat{\xi}_{s_*,k},
\end{align*}
where 
\begin{align*}
    \hat{\xi}_{s_*,k} = 
    \left\{
    \begin{array}{cc}
         L_{s_*,k} (z_{k}^{s_*} - \bold{1}\mu_{s_*,k} - \bold{1}S_{s_*,k}\hat{\eta}_{k|k-1}) & {\rm if } \ s_* \in O_k \\
         0 & {\rm otherwise}
    \end{array}
    \right.
\end{align*}
Since $\xi_{s_*,k}$ is associated with the area $s_*$, the measurement $z_{s,k}$ is not a function of $\xi_{s_*,k}$ if $s \neq s_*$. Therefore, the estimate $\hat{\xi}_{s_*,k}$ is available only when $s_* \in O_k$.
Since 
\begin{align*}
    \hat{\xi}_{s_*,k}&=L_{s_*,k}(\bold{1}S_{s_*,k}\tilde{\eta}_{k|k-1}+\epsilon_{k}^{s_*}+\bold{1}\xi_{s_*,k}),
\end{align*}
we need to choose the gain $L_{s_*,k}$ such that $L_{s_*,k}\bold{1}=\II$ for an unbiased estimate.
Then, the estimation error becomes
\begin{align}
    \tilde{\xi}_{s_*,k}&=L_{s_*,k}(\bold{1}S_{s_*,k}\tilde{\eta}_{k|k-1}+\epsilon_{k}^{s_*}).
    \label{eq:xitilde}
\end{align}
The estimation error for $q_{s_*,k}$ is given by
\begin{align}
&\tilde{q}_{s_*,k}=S_{s_*,k}\tilde{\eta}_k+\tilde{\xi}_{s_*,k}\nnum\\
&=-S_{s_*,k}K_kS_k\tilde{\eta}_{k|k-1}-S_{s_*,k}K_k(\epsilon_k+E_k\xi_k)-L_{s_*,k}\epsilon_{k}^{s_*},
\label{eq:ytilde}
\end{align}
where the relations $L_{s_*,k}\bold{1}=\II$ and~\eqref{eq:etatilde} are applied.
If $s_* \notin O_k$, then we have
\begin{align*}    &\tilde{q}_{s_*,k}=S_{s_*,k}\tilde{\eta}_k+\xi_{s_*,k}.
\end{align*}
Considering the cross relations between the error terms, we can find the covariance $P_{s_*,k}^q$ in~\eqref{eq:yhat}.

\subsection{Properties of the FRRF}\label{app:FRRFpro}
\begin{lemma}
Assume $\hat{\eta}_0$ is an unbiased estimate. The estimates $\hat{\eta}_k$, $\hat{d}_{k-1}$ and $\hat{q}_k$ are the best linear unbiased estimates (BLUE), if the gains $M_k$, $K_k$, and $L_{s_*,k}$ are chosen by
~\eqref{eq:M},~\eqref{eq:K}, and~\eqref{eq:G}, respectively.
\label{lemma:BLUE}
\end{lemma}
\begin{proof}
Assume $\hat{\eta}_{k-1}$ is unbiased. The prediction error is given by
\begin{align*}
    &z_k -\mu_{k} - S_{k}H_{k-1}\hat{\eta}_{k-1}\nnum\\
    &=S_kd_{k-1}+(\epsilon_k+ E_k\xi_k+S_k(\zeta_{k-1}+H_{k-1}\tilde{\eta}_{k-1})).
\end{align*}
By normalizing the above equation with $R_k^{-\frac{1}{2}}$, we have
\begin{align*}
    &R_k^{-\frac{1}{2}}(z_k -\mu_{k} - S_{k}H_{k-1}\hat{\eta}_{k-1})\nnum\\
    &{=}R_k^{-\frac{1}{2}}S_kd_{k-1}{+}R_k^{-\frac{1}{2}}(\epsilon_k+ E_k\xi_k+S_k(\zeta_{k-1}+H_{k-1}\tilde{\eta}_{k-1})),
\end{align*}
where the variance of the last term is normalized, i.e., $Var(R_k^{-\frac{1}{2}}(\epsilon_k+ E_k\xi_k+S_k(\zeta_{k-1}+H_{k-1}\tilde{\eta}_{k-1})))=\II$.
Now, by the Gauss Markov theorem~\cite{kailath2000linear}, we can get $M_k$ in~\eqref{eq:M}. Therefore, $\hat{d}_{k-1}$ in~\eqref{eq:dhat} is BLUE as long as $\hat{\eta}_{k-1}$ is unbiased.

Given that $\hat{\eta}_{k-1}$ and $\hat{d}_{k-1}$ are unbiased, the estimate $\hat{\eta}_k$ is unbiased
\begin{align*}
    {\mathbb E}[\tilde{\eta}_k] &= {\mathbb E}[(\II-K_kS_k)\tilde{\eta}_{k|k-1}-K_k(\epsilon_k+E_k\xi_k)]=0
\end{align*}
for any $K_k$.
Now consider the following optimization problem that minimizes the trace of the covariance $P_k^\eta$ in~\eqref{eq:estim}:
$    \min_{K_k} tr(P_{k}^\eta).$
The problem is an unconstrained convex optimization problem, and thus $K_k$ is found by taking the objective function derivative with respect to the decision variable $K_k$ and setting it equal to zero.
The solution is $K_k$ in~\eqref{eq:K}. Therefore, $\hat{\eta}_k$ is BLUE, provided that $\hat{\eta}_{k-1}$ and $\hat{d}_{k-1}$ are unbiased.

Given $\hat{\eta}_0$ is unbiased, $\hat{d}_0$ and $\hat{\eta}_1$ are BLUE by the above statements. Also, given $\hat{\eta}_{k-1}$ is unbiased (because it is BLUE), $\hat{d}_{k-1}$ and $\hat{\eta}_{k}$ are BLUE. Therefore, $\hat{\eta}_k$ and $\hat{d}_k$ are BLUE for all $k$.

Given that $\hat{\eta}_k$ is BLUE, one can show that $\hat{\xi}_k$ is BLUE by the same logic of the first paragraph in this proof. In sequel, $\hat{q}_k$ is BLUE as well. We omit its details.
\end{proof}

If we have multiple outputs in the same area, we can combine those measurements into a single output by the optimal combination considering their covariance.
So, now we assume that each region may have at most a single measurement.
Consider the following assumption.
\begin{assumption}\label{asm:Pbound}
    There exist $\bar{a}$, $\bar{h}$, $\bar{g}$, $\bar{m}$, $\underline{\xi}$ $>0$, such that the following holds for all $k \geq 0$:
    \begin{align*}
        &\|S_{s,k}\| \leq \bar{s},
        &&\|H_{k}\| \leq \bar{h},
        &&\|G_{k}\| \leq \bar{g},\\
        & \|M_{k}\| \leq \bar{m},
        &&P_k^{\eta}\geq \underline{\eta} \II.
\end{align*}
\end{assumption}

Assumption~\ref{asm:Pbound} is  widely used to show the stability of Kalman filter~\cite{reif1999stochastic,bonnabel2014contraction}.
In input-state estimation, $P_k^\eta$ is bounded, if the transformed system is uniformly observable~\cite{wan2019attack}. 

Since the measurement $z_{s,k}$ is provided by anonymous vehicles, the data center does not know when they can get a measurement $z_{s,k}$. We assume the obtainment of measurement is a random arrival process.
Poisson process is a commonly used model for random and independent message arrivals. Assumption~\ref{asm:poisson} implies that the measurement $z_{s,k}$ at area $s \in D$ is randomly  and independently obtained.
\begin{assumption}
For any $s \in D$, the measurement $z_{s,k}$ is obtained as a Poisson arrival process with arrival rate $\lambda$. The value $z_{s,k}$ is independent of the Poisson distribution.
\label{asm:poisson}
\end{assumption}

Under the assumptions mentioned above, we can show the performance of the state estimation error.

\begin{theorem}
Under Assumptions~\ref{asm:Pbound} and~\ref{asm:poisson}, the expected error $\mathbb{E}[\|\tilde{q}_k\|]$ is practically exponentially stable in probability, i.e., there exist a set of positive constants $a$, $\gamma$, and $c$ such that
\begin{align}
\mathbb{E}[\|\tilde{q}_k\|] \leq a e^{-\gamma k}+c.
\label{eq:exps}
\end{align}
\label{the:FRRFstability}
\end{theorem}
\begin{proof}
For this analysis, we reformulate the output model as follows:
\begin{align}
    z_k &= [z_{s_1},\cdots,z_{s_n}]^\top,
    \label{eq:newz}
\end{align}
which is the collection of outputs for all the areas, where $z_{s_i}=0$ if $s_i \notin O_k$.
Let us introduce an indicator matrix $\mathcal{I}_k = \diag( i_{s_1,k},\cdots,i_{s_n,k})$, where $i_{s_j,k}=1$ if $s_j \in O_k$, $i_{s_j,k}=0$ otherwise.

Consider matrices $\mathbb{M}_k$ and $\mathbb{K}_k$ with appending zeros to $M_k$ and $K_k$ such that~\eqref{eq:pred} and~\eqref{eq:estim} can be replaced with
\begin{align*}
    \hat{\eta}_{k|k-1} &= H_{k-1}\hat{\eta}_{k-1}+\mathbb{M}_k\mathcal{I}_k(z_k -\mu_{k} - S_{k}H_{k-1}\hat{\eta}_{k-1})\nnum\\
    \hat{\eta}_k &= \hat{\eta}_{k|k-1} + \mathbb{K}_k\mathcal{I}_k (z_k - \mu_{k}-S_{k}\hat{\eta}_{k|k-1}),
\end{align*}
where $\mathbb{K}_k=\mathbb{K}_k \mathcal{I}_k$ and $\mathbb{M}_k=\mathbb{M}_k \mathcal{I}_k$ hold because we've appended zeros to $K_k$ and $M_k$. Note that $z_k$ in the above equation represents all the measurements in~\eqref{eq:newz}.
Given this notation, we have the error dynamics:
\begin{align*}
\tilde{\eta}_k &= (\II-\mathbb{K}_kS_k)\tilde{\eta}_{k|k-1}-\mathbb{K}_k(\epsilon_k+E_k\xi_k)\nnum\\
&=(\II-\mathbb{K}_kS_k)(\II-G_{k-1} \mathbb{M}_kS_k)H_{k-1}\tilde{\eta}_{k-1}\nnum\\
&+(\II-\mathbb{K}_kS_k)(\II-G_{k-1} \mathbb{M}_kS_k)\zeta_{k-1}\nnum\\
&-((\II-\mathbb{K}_kS_k)G_{k-1} \mathbb{M}_k+\mathbb{K}_k)(\epsilon_k+E_k\xi_k).
\end{align*}
Choose the Lyapunov function candidate
\begin{align}
    &V_k = \tilde{\eta}_k^\top(P_k^{\eta})^{-1}\tilde{\eta}_k\nnum\\
&=\tilde{\eta}_{k-1}^\top \bar{H}_{k-1}^\top (\II-\mathbb{K}_kS_k)^\top (P_k^{\eta})^{-1}
(\II-\mathbb{K}_kS_k) \bar{H}_{k-1} \tilde{\eta}_{k-1}\nnum\\
&+2\tilde{\eta}_{k-1}^\top \bar{H}_{k-1}^\top (\II-\mathbb{K}_kS_k)^\top (P_k^{\eta})^{-1} ((\II-\mathbb{K}_kS_k) \bar{\zeta}_{k-1}\nnum\\
&+\bar{K}_k \bar{\epsilon}_k)+\bar{\zeta}_{k-1}^\top (\II-\mathbb{K}_kS_k)^\top (P_k^{\eta})^{-1} (\II-\mathbb{K}_kS_k) \bar{\zeta}_{k-1}\nnum\\
&+2\bar{\zeta}_{k-1}^\top (\II-\mathbb{K}_kS_k)^\top (P_k^{\eta})^{-1}\bar{K}_k \bar{\epsilon}_k+ \bar{\epsilon}_k^\top \bar{K}_k^\top (P_k^{\eta})^{-1} \bar{K}_k \bar{\epsilon}_k,
\label{eq:V}
\end{align}
where
\begin{align*}
    \bar{H}_k &= (\II-G_{k-1} \mathbb{M}_kS_k)H_{k-1}\nnum\\
    \bar{K}_k &=-(\II-\mathbb{K}_kS_k)G_{k-1} \mathbb{M}_k-\mathbb{K}_k\nnum\\
    \bar{\zeta}_{k-1} &=(\II-G_{k-1} \mathbb{M}_kS_k)\zeta_{k-1}.
\end{align*}
Under Assumption~\ref{asm:Pbound}, there exists $\delta \in (0,1)$ such that
\begin{align}
&\tilde{\eta}_{k-1}^\top \bar{H}_{k-1}^\top (\II-\mathbb{K}_kS_k)^\top (P_k^{\eta})^{-1}
(\II-\mathbb{K}_kS_k) \bar{H}_{k-1} \tilde{\eta}_{k-1}\nnum\\
&<\delta \tilde{\eta}_{k-1}^\top (P_{k-1}^\eta)^{-1} \tilde{\eta}_{k-1}
\label{eq:Vdi1}
\end{align}
by Claim 1.1 in~\cite{wan2019attackArx}.

Since the interarrival interval of measurements follows an exponential distribution with $\lambda$ by Assumption~\ref{asm:poisson}, we have
\begin{align*}
    \mathbb{E}[i^\alpha]&=1^\alpha\int_0^\epsilon \lambda e^{-\lambda x}dx +0^\alpha\int_\epsilon^\infty \lambda e^{-\lambda x}dx= 1-e^{-\lambda \epsilon}
\end{align*}
for some non-negative integer $\alpha \geq0$, where $\epsilon$ is a sampling interval.

The diagonal indicator matrix $\mathcal{I}_k$ satisfies
$\mathbb{E}[\mathcal{I}]=\mathbb{E}[\mathcal{I}\mathcal{I}^\top]=(1-e^{-\lambda \epsilon}) \II$ and thus
\begin{align}
    \mathbb{E}[\mathcal{I}Q \mathcal{I}^\top]
    =\mathbb{E}[\mathcal{I}\mathcal{I}^\top]\mathbb{E}[Q]
    =(1-e^{-\lambda \epsilon})\mathbb{E}[Q]
    \label{eq:expec}
\end{align}
for any independent square matrix $Q$.
Under Assumption~\ref{asm:Pbound}, there exists positive constant $c_0$ such that
\begin{align}
    &\mathbb{E}[\bar{\zeta}_{k-1}^\top (\II-\mathbb{K}_kS_k)^\top (P_k^{\eta})^{-1} (\II-\mathbb{K}_kS_k) \bar{\zeta}_{k-1}\nnum\\
    &+\bar{\epsilon}_k^\top \bar{K}_k^\top (P_k^{\eta})^{-1} \bar{K}_k \bar{\epsilon}_k]
    \leq (1-e^{-\lambda \epsilon})c_0
\label{eq:Vdi2}
\end{align}
by~\eqref{eq:expec} and Claim 1.2 in~\cite{wan2019attackArx}.
From~\eqref{eq:Vdi1} and~\eqref{eq:Vdi2}, the Lyapunov function~\eqref{eq:V} becomes
\begin{align*}
    &\mathbb{E}[V_k] \leq \delta \mathbb{E}[V_k]+(1-e^{-\lambda \epsilon})c_0 \nnum\\
    &\leq \delta^k \mathbb{E}[V_0]+\sum_{i=0}^{k-1}\delta^i (1-e^{-\lambda \epsilon})c_0 \leq \delta^k \mathbb{E}[V_0]+\frac{(1-e^{-\lambda \epsilon})c_0}{1-\delta}.
\end{align*}
Therefore, we have
\begin{align*}
    \mathbb{E}[\|\tilde{\eta}_k\|^2] \leq \frac{\bar{p}}{\underline{p}}\delta^k \mathbb{E}[\|\tilde{\eta}_0\|^2] + \frac{(1-e^{-\lambda \epsilon})c_0 \bar{p}}{1-\delta}.
\end{align*}
It follows that there exist $a_1,\gamma_1,c_1>0$ such that
\begin{align*}
    \mathbb{E}[\|\tilde{\eta}_k\|] \leq a_1 e^{-\gamma_1 k}\mathbb{E}[\|\tilde{\eta}_0\|]+c_1
\end{align*}
Since $\mathbb{E}[\|\epsilon_{s_*,k}\|]<c_2$, and  $\mathbb{E}[\|\zeta_k\|]<c_3$, by~\eqref{eq:xitilde} and~\eqref{eq:ytilde}, we have
\begin{align*}
    \mathbb{E}[\|\tilde{q}_{s_*,k}\|]&\leq
    \mathbb{E}[\|S_{s_*,k}\tilde{\eta}_k\|]+\mathbb{E}[\|\tilde{\xi}_{s_*,k}\|]\nnum\\
    &\leq
    a e^{-\gamma k}\mathbb{E}[\|\tilde{\eta}_0\|]+c
\end{align*}
for some positive constants $a$, $\gamma$, and $c$.
\end{proof}

Constant $c$ in~\eqref{eq:exps} can be seen as the expected error bound of the prior estimation, where the first term decays exponentially.

\subsection{Proofs of Lemmas~\ref{lem1} and~\ref{lem:Am}}\label{app:proof}

\subsubsection{Proof of Lemma~\ref{lem1}}
The vector $b$ in~\eqref{eq:dynamic} can be reformulated by
$b = b_m \frac{C_f}{\hat{C}_f} = b_m w$.
Therefore, for any $C_f \in [\underline{C}_f,\bar{C}_f]$, we have
\begin{align}
w \in \Omega = [\frac{\underline{C_f}}{\hat{C}_f},\frac{\bar{C_f}}{\hat{C}_f}].
\label{eq:wint}
\end{align}
The uncertainty $\theta$ in~\eqref{app:model} can be found by
\begin{align}
    \theta &= \frac{1}{2  Vw}[0,-\frac{(m+I_z)\Delta C_f}{m  \hat{C}_f}+\frac{(I_z\frac{\ell_r}{\ell_f}-m)\Delta C_r}{m  \hat{C}_f},0,\nnum\\
    &-\frac{(m+I_z)\Delta C_f}{m  \hat{C}_f}+\frac{(I_z\frac{\ell_r}{\ell_f}-m)\Delta C_r}{m  \hat{C}_f}]^\top\nnum\\
    &+ k_m (w^{-1}-1),
    \label{eq:theta_exp}
\end{align}
where $\Delta C_f \triangleq C_f-\hat{C}_f$ and $\Delta C_r \triangleq C_r-\hat{C}_r$.
Given the bounds of $C_f$ and $C_r$ in~\eqref{app:C} and that of $w$ in~\eqref{eq:wint}, 
the bounds of each element $\theta_i$ in~\eqref{eq:theta_exp}
are found by~\eqref{app:theta}, where $I_z\frac{\ell_r}{\ell_f}-m \geq 0$ in Assumption~\ref{asm1} has been applied.

Likewise, $\sigma(t)$ in~\eqref{app:model} is expressed as
\begin{align*}
    \sigma 
    &=-\frac{1}{2\hat{C}_f R}(2C_f\ell_f+C_r\ell_r(\frac{\ell_r}{\ell_f}-1)+\frac{mV^2}{2}),
\end{align*}
and its time-derivative becomes
\begin{align*}
    \dot{\sigma} &= b_m^\dagger G \frac{\partial \dot{p}^{\psi,des}}{\partial R}\dot{R}\nnum\\
    &=\frac{\dot{R}}{2\hat{C}_f R^2}(2C_f\ell_f+C_r\ell_r(\frac{\ell_r}{\ell_f}-1)+\frac{mV^2}{2}).
\end{align*}
Since $\frac{\ell_r}{\ell_f}-1 \geq 0$ by Assumption~\ref{asm1}, we have
$(2C_f\ell_f+C_r\ell_r(\frac{\ell_r}{\ell_f}-1)+\frac{mV^2}{2})>0$.
The bounds $R\geq \underline{R}$ and $|\frac{\dot{R}}{R^2}| \leq \bar{R}_d$  hold by Assumption~\ref{asm2} and lead to~\eqref{eq:bounds}.\oprocend\\

\subsubsection{Proof of Lemma~\ref{lem:Am}}
If $A_m(V) P + P A_m^\top(V)<0$ for all $V_{\min} \leq V \leq V_{\max}$, it is obvious that the same inequality holds for $V=V_{\min}$ and $V = V_{\max}$.

We prove sufficiency:
\begin{align*}
    &A_m(V) P + P A_m^\top(V) \nnum\\
    &= \alpha(V) (A_m(V_{\min}) P + P A_m^\top(V_{\min}))
    \nnum\\
    &\quad +(1-\alpha(V))(A_m(V_{\max}) P + PA_m^\top(V_{\max})) \nnum\\
    &= -\alpha(V) Q_{\min} - (1-\alpha(V))Q_{\max}.
\end{align*}
Since the right hand side is negative definite, the statement holds with $Q(V)=\alpha(V) Q_{\min} + (1-\alpha(V))Q_{\max}$.\oprocend\\

\end{document}